\documentclass[11pt]{article}%
\usepackage{amssymb}
\usepackage[onehalfspacing]{setspace}%
\usepackage{amsmath}%
\usepackage{amsfonts}%
\usepackage{graphicx}
\providecommand{\U}[1]{\protect\rule{.1in}{.1in}}
\newtheorem{theorem}{Theorem}

\newtheorem{definition}[theorem]{Definition}

\newtheorem{lemma}{Lemma}

\newtheorem{proposition}{Proposition}
\newtheorem{remark}{Remark}

\begin{document}

\title{Whataboutism\thanks{Eliaz acknowledges financial support from ISF grant
455/24. Spiegler acknowledges financial support from Leverhulme Trust grant
RPG-2023-120. We thank Tuval Danenberg and Refine.ink for helpful comments.}}
\author{Kfir Eliaz and Ran Spiegler\thanks{Eliaz: Tel-Aviv University and King's
College London. E-mail: kfire@tauex.tau.ac.il. Spiegler: Tel-Aviv University
and University College London. E-mail: rani@tauex.tau.ac.il.}}
\maketitle

\begin{abstract}
We propose a model of whataboutism --- a rhetorical strategy that deflects
criticism by citing similar misconduct that goes uncriticized on the critic's
side --- and study its implications for social norms governing offensive
speech. In an infinite-horizon psychological game with two rival camps, agents
weigh the intrinsic benefit of offensive speech against the risk of
condemnation. External criticism can be deflected via an equilibrium-based
whataboutism rebuttal. We characterize the unique dynamically stable
Psychological Subgame Perfect Equilibrium and show that the availability of
whataboutism exacerbates offensive speech, to the extent that civility norms
can break down entirely, especially in polarized societies.\bigskip
\bigskip\bigskip\bigskip\pagebreak

\end{abstract}

\section{Introduction}

A fundamental tenet of political liberalism is that even when different
factions in society have conflicting ends, they should be able to agree on the
legitimate means toward these ends. As Sunstein (2025) writes:
\textquotedblleft Liberals want people who disagree with each other to find a
way to live together, and to smile --- or at least to nod respectfully --- at
their differences.\textquotedblright\ For example, all sides should agree that
public demonstrations are acceptable as long as they do not deteriorate into
vandalism. Likewise, making sarcastic comments about a political opponent may
be fine, but humiliating her family members should be off-limits --- even if
one believes that the humiliation is an effective tactic.

In practice, this norm is often enforced in a decentralized manner, via social
sanctions that one individual imposes on another. For example, if one
politician mocks an opponent's family members, she will be condemned by
politicians or by social-media users. Such condemnations impose a social cost
on the offender. They can come from all sides of the political aisle, but
intuitively, they carry more weight when they come from the politician's own
camp. Moreover, failure to condemn inappropriate behavior by a member of one's
own camp is itself a potential cause for condemnation. Agents will take these
anticipated costs into account when choosing whether to engage in (or tacitly
condone) offensive expression.

\textit{Whataboutism} is a tactic that people sometimes use to deflect a
criticism against violations of social norms. Merriam-Webster dictionary
defines whataboutism as \textquotedblleft the act or practice of responding to
an accusation of wrongdoing by claiming that an offense committed by another
is similar or worse.\textquotedblright\ For example, when a prominent public
figure from camp A is reprimanded by a member of camp B for mocking the spouse
of a camp B politician (or for keeping silent when another prominent member of
camp A did so), a whataboutism response can take the following form:
\textquotedblleft Why are you criticizing me? Remember when a prominent member
of your party ridiculed the \textit{child} of a politician from my camp? I
don't recall seeing anybody from your party condemning her.\textquotedblright

Whatboutism is a modern-day instance of an ancient rhetorical device known as
\textit{tu quoque} --- which Merriam-Webster dictionary defines as
\textquotedblleft a retort charging an adversary with being or doing what the
adversary criticizes in others.\textquotedblright\ There is a popular
sentiment that whataboutism has become increasingly common in recent decades
(e.g., see Dykstra (2020)), and that it has sullied public discourse and
eroded the liberal ethos that we can agree over means even when we disagree
over ends. In Sucio (2025), the \textit{Forbes} journalist Peter Sucio
wrote:\footnote{https://www.forbes.com/sites/petersuciu/2025/04/04/whataboutism-has-made-civil-debate-on-social-media-nearly-impossible/}
\textquotedblleft No matter the topic -- especially those about politics and
world affairs -- it is impossible to have any civil debate on social media
without someone responding by making a counter-accusation. This isn't new, but
it has taken such a deep foothold in American politics in recent years, and it
isn't hard to see why.\textquotedblright\ Sucio went on to attribute the
spread of whataboutism to the rise of social media and political polarization.

We propose a formalization of whataboutism as an equilibrium phenomenon, and
use this formalism to explore the driving forces behind the phenomenon and its
implications for public discourse. In constructing our model, we are guided by
three desiderata. First, context should matter: Insulting a soldier on
Memorial Day is more offensive than on normal days, and whataboutism arguments
would make use of this distinction. Second, the model should enable criticism
to be directed not only at offensive speech, but also at subsequent failures
to condemn it from within the offender's camp. Finally, availability of the
whataboutism argument should be endogenous: The scenario referred to by one
camp's whataboutism argument is drawn from past speech acts performed by
agents on the other camp, which in turn obeyed the same norms shaped by
whataboutism. This means that in our model, the social cost of offensive
speech should be a function of agents' equilibrium beliefs regarding the
behavior of others.

Accordingly, we construct a model that falls into the general class of
extensive-form psychological games (Geanakopolos et al. (1989), Battigalli and
Dufwenberg (2009,2022)). Its basic outline is as follows. Society is divided
into two symmetric camps. There is a set of states. In each state, an agent
from a specific camp is selected to choose whether to perform a speech act
that gives him an idiosyncratic intrinsic payoff but offends the other camp.
The state captures the context that determines sensitivity to offensive
speech. States are ranked by sensitivity: In higher-sensitivity states, agents
are less likely to derive much satisfaction from offensive speech and more
likely to be insulted by it.\footnote{See Gift and Lastra-Anad\'{o}n (2026)
for experimental investigation of how context affects sensitivity to offensive
speech.}

An agent who considers performing or tacitly supporting an offensive speech
act weighs the intrinsic payoff it provides against the cost of being
condemned for it. The condemnation can arrive from the rival camp or from
one's own camp. Agents have repeated opportunities to voice disdain for the
offensive speech act. When agents fail to condemn one of their own, they are
themselves exposed to condemnation in the same way --- capturing statements
such as: \textquotedblleft Shame on you for keeping silent when a member of
our/your camp mocked an opponent's spouse.\textquotedblright\ The game ends as
soon as a condemnation takes place.

A condemnation's effectiveness depends on whether it is internal or external.
Condemnation by a member of one's own camp is always effective; it
automatically inflicts a cost on an agent who performed/supported offensive
speech. In contrast, external condemnation can be rebuffed with a whataboutism
argument. The argument is based on drawing a play path from the distribution
induced by equilibrium strategies, in a state that satisfies three
requirements: $(1)$ the state is at least as sensitive as the current one;
$(2)$ the initial speech-act choice is faced by an agent from the
\textit{rival} camp; and $(3)$ that agent performs an offensive speech act and
meets no internal rebuke along the play path. This definition captures
whataboutism arguments, which take the following general form:
\textquotedblleft In a different context that is at least as sensitive as the
present one, someone from your camp performed an offensive speech act and no
one from your camp condemned it.\textquotedblright\ The agent's cost from
external condemnation is defined as the probability that he cannot deflect it
with whataboutism. The fact that the cost is determined by the \textit{ex-ante
equilibrium distribution over play paths} is what turns our model into a
\textit{psychological} game.

We characterize the set of Psychological Subgame Perfect Equilibria (PSPE) in
this model. In PSPE, players follow a symmetric, stationary strategy, such
that whataboutism arguments only invoke a state of equivalent sensitivity. In
each state, there is a history-independent propensity to avoid (or internally
condemn) offensive speech. This propensity increases with the state's
sensitivity. When sensitivity is below some threshold, the propensity drops to
zero: There is a complete breakdown of norms, such that agents always engage
in offensive speech in that state, face no internal condemnation, and always
use whataboutism to deflect external criticism. A dynamic stability criterion
selects a unique PSPE: the one with the lowest amount of offensive speech
(i.e., the lowest sensitivity threshold below which norms break down). As
society becomes more polarized (in the sense that agents on both camps feel
more strongly about expression), offensive speech and the use of whataboutism
become more prevalent.

We compare PSPE to a benchmark model in which whataboutism is not allowed. In
that model, every state sees a positive fraction of agents who avoid (and
internally condemn) offensive speech. These findings are consistent with the
intuition that whataboutism erodes the civility of public discourse, and more
so in a polarized society.

\section{A Model}

There are two rival \textit{camps}, $1$ and $2$, each consisting of a measure
one of agents. We use $-i$ to denote the rival of camp $i$. There is a set of
\textit{states} $S=S_{1}\cup S_{2}$, where $S_{i}=\left\{  s_{i}^{1}%
,...,s_{i}^{n}\right\}  $ and $S_{1}\cap S_{2}=\emptyset$. The superscript
$m=1,...,n$ represents the degree of \textit{sensitivity to offensive speech}
that is associated with the state. When $m^{\prime}\geq m$, we sometimes write
$s_{i}^{m^{\prime}}\succsim s_{j}^{m}$ for any $i,j$, and say that
$s_{i}^{m^{\prime}}$ is at least as sensitive as $s_{j}^{m}$.

Denote $I(s)=i$ when $s\in S_{i}$. In state $s$, camp $I(s)$ gets an
opportunity to perform an offensive speech act targeted at the other camp. In
state $s_{i}^{m}$, the act gives \textit{intrinsic utility} $v\sim U[0,g_{m}]$
to a member of camp $i$, and \textit{intrinsic disutility} drawn from
$U[0,b_{m}]$ to a member of camp $-i$. Assume $1<g_{n}<\cdots<g_{1}<2$ and
$1<b_{1}<\cdots<b_{n}<2$. Thus, offensive speech in a more sensitive state
generates lower intrinsic utility for offenders and higher disutility for victims.

We now describe an infinite-horizon, sequential-move psychological game in
which actions are speech acts. In stage $0$, a state $s$ is publicly drawn
from the uniform distribution over $S$. In stage $1$, an agent from camp
$I(s)$ is randomly selected, privately observes his $v$, and chooses a
publicly observed action $a\in\{0,1\}$, where $a=1$ means that the agent
performs an offensive speech act (often referred to as \textquotedblleft the
act\textquotedblright). At every subsequent stage $k>1$, an agent is randomly
and uniformly selected from \textit{either }camp; his valuation is newly and
independently drawn. Agents' action set remains $\{0,1\}$, but now $a=1$ means
\textit{supporting} previous agents, while $a=0$ means \textit{condemning}
them. The game is terminated as soon as an agent chooses $a=0$.

Agents from camp $-I(s)$ follow an exogenous rule (which we endogenize below):
Condemning camp $I(s)$ with probability $\lambda b_{s}$, where $\lambda
\in(0,\frac{1}{2})$. The only endogenous objects are the strategies of camp
$I(s)$ agents. Since agents from the rival camp $-I(s)$ are selected at each
stage with probability $\frac{1}{2}$ and terminate the game with probability
$\lambda b_{s}$, it follows that any profile $\sigma$ of strategies by camp
$I(s)$ agents induces a well-defined probability distribution over finite
terminal histories. For every $s\in S$, let $\alpha_{\sigma}(s)>0$ denote the
ex-ante probability that the agent who terminates the game in $s$ is from camp
$-I(s)$, under the strategy profile $\sigma$.

To complete the game's description, we define the agents' \textit{payoffs}. As
we will see, preferences are not entirely consequentialistic. Rather, they fit
the framework of psychological game theory:

\begin{itemize}
\item When an agent in camp $I(s)$ plays $a=0$ in any stage, his payoff is
normalized to $0$.

\item When the agent plays $a=1$ and his immediate successor also plays $a=1$,
his payoff is $v$. The interpretation is that the agent derives an intrinsic
utility from performing/supporting the offensive speech act against the rival
camp, and incurs no social cost because he does not face immediate condemnation.

\item When the agent plays $a=1$ and his immediate successor is from his own
camp $I(s)$ and plays $a=0$, the agent's payoff is $v-1$. The interpretation
is that in addition to the intrinsic \textquotedblleft warm
glow\textquotedblright\ from performing/supporting the offensive speech act,
the agent experiences a social sanction in the form of immediate condemnation
from his own camp.

\item Finally (this is where the psychological-game component enters), when
the agent plays $a=1$ and his immediate successor is from the rival camp
$-I(s)$ and plays $a=0$, the agent's payoff is defined relative to a fixed
strategy profile $\sigma$. Specifically, his payoff is $v$ if he can mount a
\textquotedblleft\textit{rebuttal}\textquotedblright, and $v-1$ if he cannot.
A rebuttal consists of a selection of a state $s^{\prime}\succsim s$ for which
$I(s^{\prime})=-I(s)$, and a terminal history in that state --- drawn randomly
from the distribution over terminal histories induced by $\sigma$ --- such
that according to $\sigma,$ the stage-$1$ agent in that history plays $a=1$,
and all subsequent agents from camp $I(s^{\prime})$ play $a=1$ in that
terminal history.\bigskip
\end{itemize}

The game form captures in stylized fashion the following situation. A member
from one camp has an opportunity (in a mainstream-media interview, or on
social media) to verbally attack the rival camp. This agent may derive
intrinsic satisfaction from this attack. Alternatively, the attack may advance
the camp's goal, or energize \textquotedblleft the base\textquotedblright.
When the agent seizes the opportunity and carries out the verbal attack, there
are repeated subsequent opportunities for members from either camp to condemn
the attack or tacitly condone it by keeping silent. Like the original attack,
these opportunities can involve interviews on traditional media or
high-visibility social media posts. The assumption that $g_{m}>1$ for every
$m$ means that there is a positive measure of \textquotedblleft%
\textit{fanatics}\textquotedblright, who are willing to perform/support
offensive speech acts against the rival camp, regardless of the state's
sensitivity, even if they face certain condemnation.

The interpretation of the psychological-game component is as follows. When an
agent from camp $-I(s)$ criticizes an agent from camp $I(s)$ for
performing/supporting an offensive speech act, the latter agent can employ a
\textquotedblleft whataboutism\textquotedblright\ counterargument: In
circumstances that are at least as sensitive as $s$, an agent from the
critic's own camp engaged in offensive speech and faced no subsequent internal
rebuke. The psychological-game component arises because the rebuttal is based
on a \textit{random draw} from the distribution over play paths induced by the
strategy profile $\sigma$. The rebutting agent needs to \textit{retrieve a
concrete scenario} that would arm him with the whataboutism argument. In this
scenario, the state needs to be at least as sensitive as $s$; an agent from
the rival camp $-I(s)$ needs to perform an offensive speech act; and in the
ensuing play path, there needs to be no internal condemnation, such that the
agent who terminates the game is from camp $I(s)$.

Our modeling of the agent's retrieval process is based on two desiderata: We
want to assume that the agent has some control over the retrieval, but also
that success is uncertain. Consequently, our model mixes elements of
\textquotedblleft directed\textquotedblright\ and \textquotedblleft
random\textquotedblright\ search, in analogy to the macro-theory literature
(see Wright et al. (2021)). The agent deliberately targets a specific state,
but then he --- or someone from his own camp who argues on his behalf ---
randomly retrieves from memory a concrete state-specific scenario. Thus, when
the agent plays $a=1$, he does not know for sure whether he will be able to
successfully use whataboutism to rebuff out-group condemnation, but his choice
of which state to target is optimal in this regard.\bigskip

\noindent\textit{The interpretation of }$a$\textit{\ as a speech act}

\noindent Why do we interpret $a$ as a speech act, rather than as a physical
action? Imagine that the whataboutism rebuttal were infeasible. In this case,
we would be able to interpret $a=1$ as a physical action that benefits one's
own camp and harms the rival camp. An agent who takes this action incurs a
cost if his immediate successor plays $a=0$. This is a fairly standard
\textquotedblleft network externality\textquotedblright: The agent's benefit
from $a=1$ depends on whether other agents take this action, too.

The whataboutism rebuttal, which enables an agent to save the cost, is hard to
interpret except as a verbal argument that has to do with condemnations of bad
actions. Thus, the feasibility of whataboutism makes it more natural to
interpret actions at stages $k>1$ as verbal acts of condemnation and (explicit
or tacit) support. In principle, agents' choices at stage $k=1$ $could$ be
interpreted as physical actions as well as speech acts. Given our assumption
that the payoff structure is stationary and does not distinguish between $k=1$
and $k>1$, we find it more natural to restrict attention to the speech-act
interpretation of actions at all stages of the game.\bigskip

\noindent\textit{\textquotedblleft Microfounding\textquotedblright\ the
probability of external criticism}

\noindent The rule that determines the probability of external condemnation
can be endognized as follows. In state $s$, members of camp $-I(s)$ incur a
stochastic cost $c$ when they condemn an offensive speech act performed or
supported by members of camp $I(s)$. The cost captures the difficulty of
reaching out from one camp to another. Assume that $c\sim U[0,\bar{c}]$, where
$\bar{c}>b_{n}.$ Assume that a member of camp $-I(s)$ condemns the verbal
attack from camp $I(s)$ if and only if his disutility from the attack exceeds
the cost $c$. Recall that in state $s$, $u\sim\lbrack0,b_{s}]$. Therefore, the
aggregate probability that camp $-I(s)$ will condemn offensive speech by camp
$I(s)$ is
\[
\frac{1}{\bar{c}b_{s}}\int_{0}^{b_{s}}\int_{0}^{u}dcdu=\frac{b_{s}}{2\bar{c}%
}=\lambda b_{s}
\]
where $\lambda=1/2\bar{c}$. Note that $\lambda<\frac{1}{2}$ since $\bar
{c}>b_{n}>1.$

\section{A Benchmark}

We begin our analysis with a benchmark model in which agents \textit{cannot}
use whataboutism to rebuff criticisms. In this case, an agent who performs an
offensive speech act incurs a cost of $1$ whenever his immediate successor
(regardless of his affiliation) issues a condemnation. Payoffs in this variant
are entirely consequentialistic, and thus do not require the
psychological-games machinery. Moreover, in this variant, analysis can be
conducted for each state in isolation, whereas the main version of our model
creates a linkage across states via the whataboutism counterargument. This
comparison will enable us to explore the role of whataboutism in shaping the
norms of public discourse.

In what follows we assume that when an agent of camp $I(s)$ is indifferent
between $a=0$ and $a=1,$ he chooses $a=1.$\bigskip

\begin{proposition}
\label{prop benchmark}There is a unique Subgame Perfect Equilibrium in the
benchmark model. In each state $s$, each camp $I(s)$ agent plays $a=1$ if and
only if his intrinsic payoff $v$ is above the cutoff%
\begin{equation}
v_{s}^{\ast}=\frac{\lambda g_{s}b_{s}}{2g_{s}-1}\label{benchmark cutoff}%
\end{equation}
\bigskip
\end{proposition}

\noindent\textbf{Proof. }Since we can analyze equilibrium behavior state by
state, we will omit the $s$ subscripts. The proof proceeds in three
steps.\medskip

\noindent\textbf{Step 1. }\textit{At each stage, agents follow a cutoff rule}.

\noindent We show that for every $k>0$, there is a cutoff $v^{k}\in
\lbrack0,g]$ such that a member of camp $I(s)$ plays $a=1$ if and only if
$v\geq v^{k}$. To see why, recall that agents do not observe past agents'
valuations. Stage-$k$ agent's belief regarding the behavior of stage-$(k+1)$
agent is independent of agent $k$'s valuation. Therefore, if a stage-$k$ agent
with valuation $v$ plays $a=1$ ($a=0$), then the same action is also optimal
for a agent with valuation $v^{\prime}>v$ ($v^{\prime}<v$). $\square$\medskip

\noindent\textbf{Step 2. }$v^{k}\in(0,g)$ \textit{for every} $k$.

\noindent Suppose $v^{k}=0$. This means that camp $I(s)$ agents with
$v<\frac{1}{2}\lambda b$ play $a=1$. Such agents earn a negative payoff,
regardless of $v^{k+1}$, a contradiction. Now suppose $v^{k}=g$. This means
that camp $I(s)$ agents with $v>1$ play $a=0$, even though playing $a=1$ would
generate a strictly positive payoff, a contradiction. $\square$\medskip

It follows that in SPE, there is a sequence of interior cutoffs $(v^{1}%
,v^{2},...)$ such that for every $k$,%
\begin{equation}
v^{k}=\frac{1}{2}\lambda b+\frac{1}{2}\frac{v^{k+1}}{g}\label{vk benchmark}%
\end{equation}
\medskip\noindent\textbf{Step 3. }\textit{The only solution to the system of
equations given by (\ref{vk benchmark}) is stationary and given by
(\ref{benchmark cutoff}). }

\noindent Rearrange (\ref{vk benchmark}) as follows:%
\begin{equation}
v^{k+1}-v^{k}=(2g-1)v^{k}-\lambda bg\label{vk rearranged}%
\end{equation}
Note that $2g-1>0$. If $v^{k+1}-v^{k}>0$ $(<0$) for some $k$, then the
sequence $v^{k},v^{k+1},v^{k+2},...$ is a divergent increasing (decreasing )
sequence, which is a contradiction because $v$ is bounded between $(0,g)$. It
follows that the only solution to (\ref{benchmark cutoff}) is $v^{k+1}=v^{k}$
for every $k$, which yields (\ref{benchmark cutoff}). $\square$\medskip

Note that since $\lambda<\frac{1}{2}$, $g>1$ and $b<2$, the R.H.S. of
(\ref{benchmark cutoff}) is below $g$, hence $v^{\ast}$ is indeed an interior
cutoff given by an indifference condition. $\blacksquare$\bigskip

The equilibrium fraction of camp $I(s)$ agents who do not perform/support
offensive speech in state $s$ is%
\begin{equation}
x_{s}=\frac{\lambda b_{s}}{2g_{s}-1}\label{benchmark x}%
\end{equation}
This fraction is always in $(0,1)$. The fact that $x_{s}>0$ for every $s$
means that norms of public discourse never break down entirely --- i.e., there
is always a positive fraction of agents who refrain from offensive speech or
criticize it internally. Note that $x(s)$ is \textit{increasing} in the
parameter $b_{s}$ that determines the disutility from offensive speech that
members of camp $-I(s)$ incur. It is also \textit{increasing} in the parameter
$\lambda$ that affects camp $I(s)$ agents' exposure to external criticism.
Finally, it is \textit{decreasing} in the parameter $g_{s}$ that determines
the intrinsic utility that camp $I(s)$ agents derive from verbally attacking
the rival camp.

Note that the R.H.S. of (\ref{vk benchmark}) can be interpreted as a standard
social-preference effect: An agent trades off the intrinsic utility from bad
behavior (offensive speech) and the social cost in the form of frequency with
which he is condemned for such behavior. Thus, the benchmark model is
essentially a standard model of pro- or anti-social behavior in the presence
of social preferences.

\section{Analysis: Whataboutism}

This section characterizes equilibrium behavior in the main model. As a
preliminary step, we derive an expression for agents' psychological expected
payoffs. Fix a strategy profile $\sigma$ and consider a state $s$. Consider a
camp $I(s)$ agent with valuation $v$, who is selected to move in some stage
$k$. If the agent chooses $a=0$ and thus terminates the game, his payoff is by
definition $0$. Suppose now that the agent chooses $a=1$. Let $\beta_{\sigma}$
denote the probability that the stage-$(k+1)$ agent will choose $a=0$, given
the history, conditional on being from the same camp $I(s)$. Recall that
$\alpha_{\sigma}(s^{\prime})$ is the ex-ante probability (induced by $\sigma$)
that the agent who terminates the game in $s^{\prime}$ is from camp
$-I(s^{\prime})$. This is the probability that a randomly sampled
state-$s^{\prime}$ play path will support a whataboutism rebuttal by a member
of camp $-I(s^{\prime})$. Then, the stage-$k$ agent's \textit{psychological
expected payoff} given $\sigma$ from playing $a=1$ is%
\begin{equation}
v-\frac{1}{2}\beta_{\sigma}-\frac{1}{2}\lambda b_{s}\cdot\min_{s^{\prime
}\succsim s\mid I(s^{\prime})=-I(s)}(1-\alpha_{\sigma}(s^{\prime
}))\label{payoff from a=1}%
\end{equation}

The first term in this expression is the agent's intrinsic payoff from
performing/supporting the speech act. The second term represents the
probability that the agent faces condemnation from within his own camp, which
by assumption is impossible to deflect. The third term represents the
probability that the agent faces external condemnation which he cannot rebut
with a whataboutism argument. Note that the minimum operator is due to the
agent's ability to select among the eligible states (i.e., those that are at
least as sensitive as $s$) to maximize the probability of being able to say:
\textquotedblleft But what about state $s^{\prime}$? An agent from your camp
performed an offensive speech act in that state, which is at least as
sensitive as $s$, and I didn't see anyone from your camp criticizing
him.\textquotedblright

Note that the final term is calculated with respect to the strategy profile
$\sigma$. As a result, the agent's expected payoff is defined in terms of the
\textit{ex-ante} strategy profile, rather than the continuation strategy given
the history at which the agent moves. This is a hallmark of psychological-game
expected payoff calculations.

We now introduce our notion of equilibrium, due to Geanakopolos et al.
(1989).\bigskip

\begin{definition}
A strategy profile $\sigma$ is a Psychological Subgame Perfect Equilibrium
(PSPE) if at every history, each agent chooses $a=0$ if and only if it
maximizes his psychological expected payoff given $\sigma$.\bigskip
\end{definition}

The sequential rationality requirement is exactly as in conventional SPE; the
only difference is the dependence of agents' expected payoff from $a=1$ on the
ex-ante strategy profile, which includes behavior in other states.\bigskip

\noindent\textit{A two-state example}

\noindent Let $S=\{s_{1},s_{2}\}$ where $s_{1}\sim s_{2}$, such that both
states are associated with the same $(g,b)$. We focus on stationary and
symmetric PSPE --- i.e., there is a cutoff $v^{\ast}$ such that at every $s$,
a camp $I(s)$ agent with valuation $v$ chooses $a=1$ at any history, if and
only if $v\geq v^{\ast}$. In such an equilibrium, $\beta$ (the probability
that a camp $I(s)$ agent plays $a=0$) is given by $v^{\ast}/g$, independently
of the history. Therefore,
\begin{equation}
\alpha_{\sigma}(s)=\left(  1-\frac{v^{\ast}}{g}\right)  \cdot\frac{\frac{1}%
{2}\lambda b}{\frac{1}{2}\lambda b+\frac{1}{2}\frac{v^{\ast}}{g}}\label{alpha}%
\end{equation}
The second term is the conditional probability that given a terminal history,
the last move is by an agent from camp $-I(s)$. This probability is calculated
according to each camp's stationary probability of playing $a=0$. Note that
the expression for $\alpha_{\sigma}(s)$ is \textit{non-linear} in the cutoff
$v^{\ast}$ --- unlike the benchmark case. This non-linearity reflects the fact
that the whataboutism argument is based on the distribution over play paths
induced by agents' ex-ante strategies, and it will play an important role in
the sequel.

Rearranging (\ref{alpha}), we obtain%
\[
1-\alpha_{\sigma}(s)=\frac{v^{\ast}+b\lambda v^{\ast}}{v^{\ast}+b\lambda g}
\]

The cutoff $v^{\ast}$ equates the agent's payoff from $a=1$ and $a=0$:%
\[
v^{\ast}-\frac{1}{2}\frac{v^{\ast}}{g}-\frac{1}{2}\lambda b\frac{v^{\ast
}+b\lambda v^{\ast}}{v^{\ast}+b\lambda g}=0
\]
This equation has potentially two solutions. One solution is $v^{\ast}=0$,
which corresponds to a stationary equilibrium in which all agents always play
$a=1$ \textit{regardless} of their valuation. When $g<1+\frac{1}{2}b\lambda$,
there is an additional solution:%
\begin{equation}
v^{\ast}=\frac{\lambda bg}{2g-1}\left(  \lambda b-2g+2\right)
\label{cutoff example}%
\end{equation}

The fraction of camp $I(s)$ agents who play $a=0$ in state $s$ at any stage is
$x^{\ast}=v^{\ast}/g$. Note that since $g>1$ and $b\lambda<1$, this fraction
$x^{\ast}$ is equal to the benchmark level (\ref{benchmark x}) multiplied by a
\textquotedblleft discount factor\textquotedblright\ $b\lambda-2g+2<1$. This
factor captures an equilibrium effect of whataboutism, which raises the
fraction of agents who perform/support offensive speech acts. When
$b\lambda-2g+2<0$ there is no interior solution, and $x^{\ast}=0$ in the
unique equilibrium --- i.e., all norms that constrain public expression break down.

\subsection{The Main Result}

In this subsection we characterize the PSPE and show that there is a unique
stable equilibrium. The following notation will be useful:%
\begin{align}
c_{m}  & =\frac{\lambda b_{m}}{2g_{m}-1}\label{c theta phi}\\
\theta_{m}  & =\lambda b_{m}-2g_{m}+2\nonumber\\
\phi_{m}(z)  & =\frac{\lambda b_{m}z+z}{\lambda b_{m}+z}\nonumber
\end{align}
Note that $c_{m}$ is equal to the fraction of agents who play $a=0$ in SPE
under the benchmark model. Note also that $\theta_{m}<1$ (since $\lambda
b_{m}<1$ and $g_{m}>1$) and increases in $m$. Let $M$ be the lowest value of
$m$ for which $\theta_{m}\geq0$.\bigskip

\begin{proposition}
\label{prop main result}For every $m^{\ast}\in\{M,...,n+1\}$, there is a PSPE
such that: (i) for every $m<m^{\ast}$, camp $I(s_{i}^{m})$ agents always play
$a=1$ in state $s_{i}^{m}$; (ii) for every $m\geq m^{\ast}$, a camp
$I(s_{i}^{m})$ agent with valuation $v$ plays $a=1$ if and only if $v\geq
v_{m}^{\ast}=c_{m}\theta_{m}g_{m}$; and (iii) the whataboutism rebuttal in any
state $s_{i}^{m}$ selects the state $s_{-i}^{m}$. There are no other
equilibria.\bigskip
\end{proposition}

\noindent\textbf{Proof}. The proof proceeds stepwise.\bigskip

\noindent\textbf{Step 1}: \textit{In PSPE, agents follow a stationary cutoff
strategy in each state.}

\noindent By essentially the same arguments as in Step 1 of the proof of
Proposition \ref{prop benchmark}, for very state $s$ and every stage $k$,
there is a cutoff $v_{s}^{k}$ such that a camp $I(s)$ agent with valuation $v
$ who moves in stage $k$ plays $a=1$ if and only if $v\geq v_{s}^{k}$.

Given a PSPE $\sigma,$ $\mu_{s}\left(  \sigma\right)  :=\min_{s^{\prime
}\succsim s\mid I(s^{\prime})=-I(s)}(1-\alpha_{\sigma}(s^{\prime}))$ is the
probability that external condemnation of camp $I(s)$ in state $s$ cannot be
rebuffed by whataboutism. Then, in state $s$, a camp $I(s)$ agent with
valuation $v$ who moves in stage $k$ earns a payoff of%
\begin{equation}
v-\frac{1}{2}\lambda b_{s}\mu_{s}\left(  \sigma\right)  -\frac{1}{2}%
\frac{v_{s}^{k+1}}{g_{s}}\label{eq payoff}%
\end{equation}
from $a=1$, and a payoff of zero from $a=0$. If $\mu_{s}\left(  \sigma\right)
>0$, then by the same arguments as in steps 2 and 3 of the proof of
Proposition \ref{prop benchmark}, $v_{s}^{k}=v_{s}^{k+1}$ for every $k$. If
$\mu_{\sigma}(s)=0$, then by the definition of the cutoff, $v_{s}^{k+1}%
=2g_{s}v_{s}^{k}$, implying that the sequence $v_{s}^{k},v_{s}^{k+1},...$ is a
divergent increasing sequence. But this violates the assumption that $v\leq
g_{s}$. $\square$\bigskip

Let $v_{s}^{\ast}$ denote the cutoff that characterizes state $s$. Fix
$i=1,2$. When $s=s_{i}^{m}$ ($s=s_{-i}^{m}$), we use $x_{m}=v_{s}^{\ast}%
/g_{s}$ ($y_{m}=v_{s}^{\ast}/g_{s}$) to denote the fraction of members of camp
$i$ ($-i$) who play $a=0$ in equilibrium in state $s$.\bigskip

\noindent\textbf{Step 2}: \textit{In every state }$s_{i}^{m}$\textit{,}%
\begin{equation}
\mu_{s_{i}^{m}}\left(  \sigma\right)  =\min_{m^{\prime}\geq m}\phi_{m^{\prime
}}(y_{m^{\prime}})\label{gamma}%
\end{equation}

\noindent By Step 1, in equilibrium, every state $s$ is associated with a
stationary cutoff $v_{s}^{\ast}$ that governs camp $I(s)$ agents' behavior.
Therefore, for every state $s^{\prime}$ where $I(s^{\prime})\neq I(s)$, we
obtain essentially the same expression as (\ref{alpha}), repeated here for
convenience:%
\[
\alpha_{\sigma}(s^{\prime})=\left(  1-\frac{v_{s^{\prime}}^{\ast}%
}{g_{s^{\prime}}}\right)  \cdot\frac{\frac{1}{2}\lambda b_{s^{\prime}}}%
{\frac{1}{2}\lambda b_{s^{\prime}}+\frac{1}{2}\frac{v_{s^{\prime}}^{\ast}%
}{g_{s^{\prime}}}}
\]
It follows that for $s^{\prime}=s_{-i}^{m^{\prime}},$ the probability
$1-\alpha_{\sigma}(s^{\prime})$ simplifies to the expression $\left(  \lambda
b_{m^{\prime}}y_{m^{\prime}}+y_{m^{\prime}}\right)  /\left(  \lambda
b_{m^{\prime}}+y_{m^{\prime}}\right)  ,$ which by the definition of $\mu
_{s}\left(  \sigma\right)  $ yields equation (\ref{gamma}). $\square$\bigskip

\noindent\textbf{Step 3}: PSPE can be characterized as follows: \textit{For
every }$m=1,...,n$,%
\begin{align}
x_{m}  & =c_{m}\min_{m^{\prime}\geq m}\phi_{m^{\prime}}(y_{m^{\prime}%
})\label{x y pspe def}\\
y_{m}  & =c_{m}\min_{m^{\prime}\geq m}\phi_{m^{\prime}}(x_{m^{\prime}%
})\nonumber
\end{align}
\bigskip

\noindent Consider a state $s=s_{i}^{m}.$ From Step 1 and equation
(\ref{eq payoff}) it follows that the equilibrium cutoff at this state
satisfies $v_{s}^{\ast}=\frac{1}{2}\lambda b_{s}\mu_{s}\left(  \sigma\right)
+\frac{1}{2}(v_{s}^{\ast}/g_{s})$. Dividing both sides of this equation by
$g_{s}$, replacing $(b_{s},g_{s})$ with $(b_{m},g_{m})$ and rearranging
yields
\[
x_{m}=\frac{\lambda b_{m}}{2g_{m}-1}\mu_{s_{i}^{m}}\left(  \sigma\right)
\]
which coincides with the stated equation for $x_{m}$ by the definitions of
$c_{m}$ and Step 3. An analogous computation for a state $s_{-i}^{m}$ yields
the desired equation for $y_{m}.$ $\square$ \bigskip

The final step of the proof uses the following lemmas.\bigskip

\begin{lemma}
\label{lemma phi}For every $m=1,...,n,$ and every $z>0$: $(i)$ $\phi_{m}%
(c_{m}z)>z$ if $z<\theta_{m}$, and $\phi_{m}(c_{m}z)<z$ if $z>\theta_{m}$; and
$(ii)$ $\phi_{m}(c_{m}z)=z$ has the unique solution $z=\theta_{m}.$\bigskip
\end{lemma}

\noindent\textbf{Proof.} Note that%
\[
\phi_{m}(c_{m}z)-z=\frac{\left(  1+\lambda b_{m}\right)  c_{m}z}{\lambda
b_{m}+c_{m}z}-z=\frac{c_{m}z\left[  \left(  \lambda b_{m}-2g_{m}+2\right)
-z\right]  }{\lambda b_{m}+c_{m}z}
\]
Since $\lambda b_{m}+c_{m}z>0$ and $c_{m}z>0$, the result follows. $\square$
\bigskip

\begin{lemma}
\label{lemma uniqueness}Suppose there exist $m\in\left\{  1,...,n\right\}  $
and $T\geq\theta_{m}$ such that%
\begin{align}
x_{m}  & =c_{m}\min\left\{  \phi_{m}(y_{m}),T\right\} \label{eq lemma unique}%
\\
y_{m}  & =c_{m}\min\left\{  \phi_{m}(x_{m}),T\right\} \nonumber
\end{align}
If $\theta_{m}>0$, then either $x_{m}=y_{m}=0$ or $x_{m}=y_{m}=c_{m}\theta
_{m}$ (and in the latter case, $\phi_{m}(x_{m})=\phi_{m}(y_{m})=\theta_{m} $).
If $\theta_{m}\leq0$, then $x_{m}=y_{m}=0$.\bigskip
\end{lemma}

\noindent\textbf{Proof.} Let $(x_{m},y_{m})$ be the solutions to
(\ref{eq lemma unique}). Since $\phi_{m}\left(  \cdot\right)  $ is strictly
increasing on $[0,\infty),$ $x_{m}>y_{m}$ implies $\phi_{m}\left(
x_{m}\right)  >\phi_{m}(y_{m})$ and hence, $\min\left\{  \phi_{m}%
(x_{m}),T\right\}  \geq\min\left\{  \phi_{m}(y_{m}),T\right\}  $, which
implies $y_{m}\geq x_{m},$ a contradiction. Similarly, $y_{m}>x_{m}$ is
impossible. Therefore, $x_{m}=y_{m}=t$, such that $t=c_{m}\min\left\{
\phi_{m}(t),T\right\}  $. Note that since $\phi_{m}(0)=0$, $t=0$ solves this equation.

Let us now look for strictly positive solutions for $t$. Suppose $\phi
_{m}(t)>T$. Then, $t=c_{m}T$, such that $\phi_{m}(c_{m}T)>T.$ By Lemma
\ref{lemma phi}, this holds if and only if $T<\theta_{m}$, a contradiction.
Therefore, $\phi_{m}(t)\leq T$, such that $t=c_{m}\phi_{m}(t)$. For $t>0$,
this implies $1=c_{m}\left(  \lambda b_{m}+1\right)  /(\lambda b_{m}+t)$, such
that $t=c_{m}\theta_{m}$, and therefore $\phi_{m}(t)=\theta_{m}$. By
assumption, $\theta_{m}\leq T$, hence the solution is consistent with the
guess that $\phi_{m}(t)\leq T$. We conclude that a strictly positive solution
to (\ref{eq lemma unique}) exists iff $\theta_{m}>0,$ and when it exists it is
equal to $c_{m}\theta_{m}.$ $\square$\bigskip

\noindent\textbf{Step 4}: In PSPE, there is $m^{\ast}\in\{M,...,n+1\}$, such
that $x_{m}=y_{m}=0$\textit{\ for }$m<m^{\ast},$\textit{\ while }$x_{m}%
=y_{m}=c_{m}\theta_{m}$\textit{\ and }$\phi_{m}(x_{m})=\phi_{m}(y_{m}%
)=\theta_{m}$\textit{\ for }$m\geq m^{\ast}$.

\noindent First, suppose that in some PSPE, $x_{m}=0$ or $y_{m}=0$ for some
$m$. Then, since $\phi_{m^{\prime}}(0)=0$ for every $m^{\prime}$, Step 3
implies that $x_{m^{\prime}}=y_{m^{\prime}}=0$ for every $m^{\prime}\leq m $.
Thus, in PSPE, there is $m^{\ast}\in\{1,...,n+1\}$ such that $x_{m}=y_{m}=0$
for every $m<m^{\ast}$, and $x_{m},y_{m}>0$ for every $m\geq m^{\ast}$. If
$m^{\ast}=n+1$, we are done. Now suppose $m^{\ast}\leq n$. The proof proceeds
by backward induction on $m.$

\noindent\textit{Base case: }$m=n$\textit{. }Clearly, $\min_{m\geq n}\phi
_{m}(y_{m})=\phi_{n}(y_{n})$. Apply Lemma \ref{lemma uniqueness} with
$T=+\infty$. If $\theta_{n}\leq0$, we have $x_{n}=y_{n}=0$, a contradiction.
Therefore, from now on we must assume $\theta_{n}>0$. In this case, the lemma
implies $x_{n}=y_{n}=c_{n}\theta_{n}$ and $\phi_{n}(x_{n})=\phi_{n}%
(y_{n})=\theta_{n}$.

\noindent\textit{Inductive step. }Fix some $m\in\{m^{\ast},...,n\}$, and
assume the claim holds for $m+1,...,n+1$. By the inductive step,
$x_{m^{\prime}}=y_{m^{\prime}}=c_{m^{\prime}}\theta_{m^{\prime}}$ and
$\phi_{m^{\prime}}(x_{m^{\prime}})=\phi_{m^{\prime}}(y_{m^{\prime}}%
)=\theta_{m^{\prime}}$ for al $m^{\prime}>m$. Note that $\theta_{m}$ is
increasing in $m$. Therefore,
\[
\min_{m^{\prime}>m}\phi_{m^{\prime}}(x_{m^{\prime}})=\min_{m^{\prime}>m}%
\phi_{m^{\prime}}(y_{m^{\prime}})=\theta_{m+1}
\]
Applying Lemma \ref{eq lemma unique} with $T=\theta_{m+1}\geq\theta_{m}>0$, we
obtain $x_{m}=y_{m}=c_{m}\theta_{m}$ and $\phi_{m}(x_{m})=\phi_{m}%
(y_{m})=\theta_{m}$. This completes the inductive proof. $\square$\bigskip

It is now straightforward to confirm that the sequences $(x_{m})$ and
$(y_{m})$ described by Step 4 solve (\ref{x y pspe def}). Therefore, Step 4
fully characterizes the set of PSPE. It also means that for any state
$s_{i}^{m}$, the state $s^{\prime}$ that minimizes $1-\alpha_{\sigma
}(s^{\prime})$ among eligible ones is $s_{-i}^{m}$ --- i.e., a whataboutism
rebuttal always targets an equivalent-sensitivity state. $\blacksquare
$\bigskip

Proposition \ref{prop main result} means that any PSPE is characterized by a
\textquotedblleft sensitivity threshold\textquotedblright\ $m^{\ast}\geq M$.
In state $s_{i}^{m},$ $m\geq m^{\ast}$, the probability that a camp $i$ agent
refrains from offensive speech (and condemns his comrades when they do not) is
$x_{m}=v_{m}^{\ast}/g_{m}$, which is equal to the corresponding probability in
the benchmark ($c_{m}$) multiplied by the \textquotedblleft discount
factor\textquotedblright\ $\theta_{m}.$ Since both $c_{m}$ and $\theta_{m}$
increase with $m,$ the probability of condemnation \textit{increases} with the
sensitivity of the state. The term $\theta_{m}$ can be interpreted as the
\textquotedblleft whataboutism effect\textquotedblright, which lowers the
likelihood that agents condemn members of their camp for performing/supporting
offensive speech. The whataboutism effect can be dramatic and lead to a
\textquotedblleft breakdown of norms\textquotedblright\ in the sense that in
if $\theta_{m}<0$ for some state $s_{i}^{m},$ then in all weakly less
sensitive states, \textit{all} agents (regardless of their camp)
perform/support offensive speech.

A key aspect of the result is that when an agent employs the whataboutism
rebuttal in some state $s_{i}^{m}$, he chooses to target the equally sensitive
mirror state $s_{-i}^{m}$. This feature, combined with the symmetries of the
model, leads to the symmetry of equilibrium behavior across camps, and to the
feature that the equilibrium equations treat each sensitivity level $m$ in
isolation. The logic behind this feature is that $\alpha_{\sigma}(s)$ (which
is the probability that a randomly sampled state-$s$ play path supports a
whataboutism argument by a member of camp $-I(s)$) is decreasing in the
sensitivity of $s$. Establishing this monotonicity takes the bulk of the proof
of Proposition \ref{prop main result}. Therefore, a agent who employs
whataboutism in state $s_{i}^{m}$ will target at the least sensitive eligible
state, which is therefore $s_{-i}^{m}$.

The PSPE are ranked by their amount of offensive speech: The higher $m^{\ast}%
$, the higher that amount. In particular, there is always an equilibrium with
complete breakdown of norms, where $m^{\ast}=n+1$ such that agents always
engage in (and never condemn) offensive speech. To see why, note that the
equilibrium probability that in state $s$ a camp $I(s)$ agent chooses $a=0$ is
equal to the benchmark probability $\lambda b_{s}g_{s}/(2g_{s}-1)$, multiplied
by the minimal probability that this agent will not be able to counter a
criticism from the rival camp; and if all agents always choose $a=1,$ this
probability is zero. However, as we show below, the \textit{unique dynamically
stable} equilibrium among all PSPE is the one with the least amount of
offensive speech.

The proof of Proposition \ref{prop main result} establishes that a PSPE
induces a sequence $x=(x_{m})_{m=1}^{n}$ --- where $x_{m}$ is the probability
that a camp $i$ agent chooses $a=0$ in state $s_{i}^{m}$ --- that satisfies
$x_{m}=c_{m}\phi_{m}(x_{m})$ for every $m=1,...,n$, where $c_{m}$ and
$\phi_{m}$ are given by (\ref{c theta phi}). We can regard the R.H.S. of the
equation $x_{m}=c_{m}\phi_{m}(x_{m})$ as a self-mapping from $x$ to itself,
analogous to a best-reply correspondence in standard game theory. The $x_{m}$
on the R.H.S. represents the fraction of agents on either camp who play $a=0$
in the states $s_{1}^{m}$ and $s_{2}^{m}$. The $x_{m}$ on the L.H.S. is
induced by the individual agent's cutoff $v_{m}^{\ast}$, which is his
best-response to the population-level behavior. In light of this
interpretation of the equilibrium equations, we propose the following notion
of dynamic stability.\bigskip

\begin{definition}
A PSPE is dynamically stable if its induced sequence $x=(x_{m})_{m=1}^{n}$
satisfies the following: There is $\delta>0$ such that for every $m$ and
$x_{m}^{\prime}$ satisfying $\left\vert x_{m}^{\prime}-x_{m}\right\vert
<\delta$, we have $\left\vert \phi_{m}(c_{m}x_{m}^{\prime})-\phi_{m}%
(c_{m}x_{m})\right\vert <\left\vert x_{m}^{\prime}-x_{m}\right\vert $.\bigskip
\end{definition}

Put differently, a PSPE is dynamically stable if at every state, a small
perturbation in the fraction of agents (in both camps) who choose $a=0$ in any
state leads to a smaller change in the probability that an individual agent
chooses that action.\footnote{Note that $\phi\left(  x\right)  $ is monotone
in $x$, so when $x$ changes, $\phi\left(  x\right)  $ changes in the same
direction.}\bigskip

\begin{proposition}
\label{prop stability}The unique dynamically stable PSPE has $m^{\ast}%
=M$.\bigskip
\end{proposition}

\noindent\textbf{Proof.} We begin by establishing relations that are related
to those given by Lemma \ref{lemma phi}. Suppose $\theta_{m}>0$. Denoting
$w=c_{m}z$ and applying Lemma \ref{lemma phi}, we obtain that for any $w>0$:
$c_{m}\phi_{m}(w)>w$ if $w<c_{m}\theta_{m}$; $c_{m}\phi_{m}(w)<w$ if
$w>c_{m}\theta_{m}$; and $c_{m}\phi_{m}(w)=w$ has the unique solution
$w=c_{m}\theta_{m}$. Now suppose $\theta_{m}\leq0.$ Then, by the same
argument, the only non-negative solution to $c_{m}\phi_{m}(w)=w$ is $w=0$, and
$c_{m}\phi_{m}(w)<w$ for any $w>0$.

We now show that the PSPE with $m^{\ast}=M$ is dynamically stable. By the
previous paragraph, for every $m\geq M,$ if there exists a solution $x_{m}>0$
to the equation $c_{m}\phi_{m}(x_{m})=x_{m},$ then it is unique.\ Moreover,
$x>c_{m}\phi_{m}(x)$ for every $x>x_{m}$ while $x<c_{m}\phi_{m}(x)$ for every
$x<x_{m}.$ Hence, for any such solution $x_{m}$,
\[
c_{m}\phi_{m}(x)-c_{m}\phi_{m}(x_{m})=c_{m}\phi_{m}(x)-x_{m}<x-x_{m}
\]
for any $x>x_{m}$, while%
\[
c_{m}\phi_{m}(c_{m}x)-c_{m}\phi_{m}(x_{m})=c_{m}\phi_{m}(c_{m}x)-x_{m}%
>x-x_{m}
\]
for any $x<x_{m}$.

Next consider $m<M$. By the definition of $M$ (as the lowest value of
$m^{\prime}$ for which $\theta_{m^{\prime}}>0$), and by the first paragraph,
the only non-negative solution to $c_{m}\phi_{m}(x_{m})=x_{m}$ is $x_{m}=0$,
and for any $x>0,$
\[
c_{m}\phi_{m}(x)-c_{m}\phi_{m}(x_{m})=c_{m}\phi_{m}(x)<x=x-x_{m}
\]
We have thus established the dynamic stability of the PSPE with $m^{\ast}=M$.

Finally, consider a PSPE with $m^{\ast}>M$. By the definition of $m^{\ast}$,
$x_{M}=0$. By the definition of $M,$ $\theta_{M}>0.$ Hence, by the first
paragraph, for any $x\in(0,c_{M}\theta_{M}),$%
\[
c_{M}\phi_{M}(x)-c_{M}\phi_{M}(x_{M})=c_{M}\phi_{M}(x)>x=x-x_{M}
\]
implying that the equilibrium is not dynamically stable. $\blacksquare
$\bigskip

It follows that dynamic stability uniquely selects the PSPE with the lowest
amount of offensive speech. In the two-state example we analyzed earlier in
this section, this has the following simple implication: If an equilibrium
with $v^{\ast}>0$ exists, it is the unique stable one. Otherwise, the unique
equilibrium $v^{\ast}=0$ is also stable.\bigskip

\noindent\textit{Non-linearity and the uniform distribution of }$v$

\noindent From a narrowly technical point of view, what makes the equilibrium
analysis of whataboutism distinct from the benchmark model, is that the cost
of playing $a=1$ is non-linear in the behavior of others. This non-linearity
arises from the way the success rate of a whataboutism rebuttal depends on the
equilibrium state-dependent frequency of internal condemnation.

This technical comment puts our assumption that $v$ (the intrinsic payoff from
offensive speech acts) is uniformly distributed in perspective. This
assumption is made partly for tractability, but also to make it clear that the
non-linear effects of whataboutism do not rely on distributional assumptions.
Our results would carry over to non-uniform distributions of $v$, albeit
without transparent closed-form expressions for equilibrium cutoffs.

\subsection{The Frequency of Offensive Speech and Whataboutism}

We use our characterization of the unique \textit{dynamically stable} PSPE to
analyze the equilibrium frequency of offensive speech and the use of
whataboutism. By Proposition \ref{prop main result}, the fraction of agents
who refrain from offensive speech (or condemn it internally) in the states of
sensitivity $m$ is
\begin{align}
x_{m}  & =y_{m}=\frac{v_{m}^{\ast}}{g_{m}}\label{eq x y}\\
& =\max\left\{  0,c_{m}\theta_{m}\right\} \nonumber\\
& =\max\left\{  0,\frac{\lambda b_{m}}{2g_{m}-1}\left(  \lambda b_{m}%
-2g_{m}+2\right)  \right\} \nonumber
\end{align}

The interpretation of the condition $\theta_{m}=\lambda b_{m}-2g_{m}+2>0$
(which holds for $m\geq M$) is that the probability of external condemnation
($\lambda b_{m}$) is more than twice the fraction of \textquotedblleft
fanatics\textquotedblright\ in the camp choosing whether to engage in
offensive speech ($g_{m}-1$). The following result applies comparative statics
to the expression (\ref{eq x y}).\bigskip

\begin{remark}
The equilibrium fraction of agents who choose $a=1$ in states of sensitivity
$m$ weakly increases in $g_{m}$ and weakly decreases in $\lambda$ and $b_{m} $
(strictly so when $m\geq M$).\bigskip
\end{remark}

The property that $x_{m}$ increases in $b_{m}$ and decreases in $g_{m}$ is
shared by the benchmark model. However, the availability of whataboutism
introduces a \textquotedblleft\textit{multiplier effect}\textquotedblright%
\ given by $\theta_{m}$. The same increase in $g_{m}$, say, leads to a
stronger decrease in $x_{m}$, relative to the benchmark. The intuition is as
follows. When $g_{m}$ goes up, the fraction of \textquotedblleft
fanatics\textquotedblright\ in each camp grows. This has a direct positive
effect on the propensity to play $a=1$, as in the benchmark model, because
playing $a=1$ exposes an agent to weaker internal criticism. However, unlike
the benchmark model, there is also an indirect effect: The increase in the
fraction of fanatics in the rival camp also makes it easier to deflect
external criticism with whataboutism (because it is easier to come with the
scenario that the rival camp engaged in offensive speech and faced no internal
rebuke). The combined effect is stronger than in the benchmark model.

Let us now turn to the equilibrium usage of whataboutism. The probability that
the whataboutism rebuttal occurs in state $s_{i}^{m}$ is the joint probability
of the following two events: $(1)$ the first-moving agent (who belongs to camp
$i$) chooses $a=1$, and the agent who terminates the game by condemning him is
from camp $-i$; and $(2)$ in the equally sensitive state $s_{-i}^{m},$ the
first-moving agent (who belongs to camp $-i$) chooses $a=1$, and the agent who
terminates the game by condemning him is from camp $i$. By symmetry, the
probabilities of these two events are equal and given by:%
\[
\left(  1-x_{m}\right)  \left(  \frac{\frac{1}{2}\lambda b_{m}}{\frac{1}%
{2}\lambda b_{m}+\frac{1}{2}x_{m}}\right)  =2g_{m}-\lambda b_{m}%
-1=1-\theta_{m}
\]
where the left-hand equality is obtained by plugging (\ref{eq x y}). This
implies the following observation.\bigskip

\begin{remark}
The frequency of whataboutism in states of sensitivity $m$ is $\min\{1,\left(
1-\theta_{m}\right)  ^{2}\}$. It weakly increases in $g_{m}$ and weakly
decreases in $\lambda$ and $b_{m}$ (strictly so when $m\geq M$).\bigskip
\end{remark}

Thus, as sensitivity becomes weaker, we observe more offensive speech and more
whataboutism in equilibrium. When norms break down entirely such that all
agents engage in offensive speech, whataboutism is \textit{always}
employed.\bigskip

\noindent\textit{The Equilibrium Effect of Polarization}

\noindent Modify the parameters $b$ and $g$ as follows: For every $m$, replace
$b_{m}$ and $g_{m}$ with $kb_{m}$ and $kg_{m}$, respectively, where $k\geq1$.
This captures a society that is more polarized in its attitudes to offensive
speech: The people who derive satisfaction from offensive speech feel more
strongly about it, while their targets are more likely to take offense. Note
also that the difference between the two camps expected valuations is
$k(g_{m}+b_{m})/2$ (recall that $b_{m}$ indicates a disutility). Hence,
polarization in this usual sense also increases with $k$.

Differentiate the equilibrium expressions for $x_{m}$ and the frequency of
whataboutism in state $m$ w.r.t $k$, and evaluate the derivative at $k=1$. The
derivative of $x_{m}$ is negative while the derivative of $(1-\theta_{m})^{2}$
is positive. This means that as society becomes more polarized, both offensive
speech and whataboutism become more frequent.

\section{Conclusion}

This paper has formalized the argument that whataboutism erodes civility norms
in public discourse, and that both offensive speech and the use of
whataboutism proliferate in polarized societies . We have argued that
psychological game theory offers a suitable tool for modeling whataboutism,
because the effectiveness of this rhetorical device depends on the relative
equilibrium frequencies of offensive speech acts and the reactions to it on
both sides of the political aisle. The psychological-game factor introduces a
non-linear dependence on equilibrium strategies into agents' expected payoff
calculations, which generates potential equilibrium multiplicity. However,
dynamic stability considerations select a unique equilibrium, which
nonetheless is qualitatively distinct from equilibrium in a benchmark model in
which anti-social behavior is incentivized by a standard social cost that is
linear in the behavior of others.

Although our discussion of the equilibrium effects of whataboutism has
suggested that these effects are adverse, we have refrained from explicit
utilitarian welfare analysis. Since utilitarianism favors agents who derive
sadistic pleasure from offending others or are quick to take offense from
controversial speech, we believe it is inappropriate in the present
context.

\end{document}